%% file: spiral-galaxies-arXiv.tex
\documentclass{elsarticle}

\usepackage{lineno,hyperref}
\modulolinenumbers[5]

\usepackage{microtype} 
\usepackage [usenames] {xcolor}
\usepackage {placeins}

\graphicspath{{./figures/}}

\usepackage{todonotes}
\usepackage{amssymb}
\usepackage{amsmath}

\newcommand{\Z}{\mathbb{Z}}

\newcommand{\defn}[1]{\textbf{\textit{\boldmath #1}}}

\newtheorem{theorem}{Theorem}
\newtheorem{corollary}{Corollary}
\newproof{proof}{Proof}
\newtheorem{lemma}{Lemma}
\newtheorem{observation}{Observation}
\newtheorem{definition}{Definition}

\let\realendproof=\endproof
\def\endproof{\hspace*{\fill}\qed\realendproof}

\newcommand{\proofofref}{}
\newproof{zproofof}{Proof of \proofofref}

\def\comic#1#2#3{\parbox{#1}{\centering\includegraphics[width=#1]{#2}\\\vspace*{-.15cm}{\footnotesize #3}\vspace*{.25cm}}}
\def\comicII#1#2{\parbox{#1}{\centering\includegraphics[width=#1]{#2}}}

\title{Rectangular Spiral Galaxies are Still Hard\footnote{This research was performed in part at the 30th and the 33rd Bellairs Winter Workshop
on Computational Geometry. We thank all other participants for a fruitful
atmosphere.}\hspace*{.05cm}\footnote{A preliminary extended abstract appeared as a preprint in the Proceedings of the 37th European Workshop on Computational Geometry~\cite{dls-rsgsh-21}.}}

\author{Erik D. Demaine}
\address{Computer Science and Artificial Intelligence Laboratory,\\
Massachusetts Institute of Technology\\
  \texttt{edemaine@mit.edu}}

\author{Maarten L{\"o}ffler}
\address{Department of Information and Computing Sciences, 
Universiteit Utrecht\\
  \texttt{m.loffler@uu.nl}}

\author{Christiane Schmidt}
\address{Department of Science and Technology, Link\"oping University\\
  \texttt{christiane.schmidt@liu.se}}

\begin{document}

\begin{abstract}
Spiral Galaxies is a pencil-and-paper puzzle played on a grid of unit squares: given a set of points called \emph{centers}, the goal is to partition the grid into polyominoes such that each polyomino contains exactly one center and is $180^\circ$ rotationally symmetric about its center.
We show that this puzzle is NP-complete, ASP-complete, and \#P-complete even if
(a)~all solutions to the puzzle have rectangles for polyominoes; or
(b)~the polyominoes are required to be rectangles and
all solutions to the puzzle have just
$1\times1$, $1\times3$, and $3\times1$ rectangles.
The proof for the latter variant also implies NP/ASP/\#P-completeness of finding a noncrossing perfect matching in distance-$2$ grid graphs where edges connect vertices of Euclidean distance~$2$.
Moreover, we prove NP-completeness of the design problem of minimizing the number of centers such that there exists a set of galaxies that exactly cover a given shape.
\end{abstract}

\maketitle

\section{Introduction}\label{sec:introduction}

Spiral Galaxies is a pencil-and-paper puzzle published by Nikoli since 2001~\cite{n-sg-01} under the name ``Tentai Show''~\cite{gsp-18}. It is played on a grid of unit squares with given ``centers'' --- points located at grid vertices, cell centers, or edge midpoints. The goal is to decompose the grid into polyominoes called ``galaxies'' such that each galaxy contains exactly one center and is $180^\circ$ rotationally symmetric about its center; see Figure~\ref{fig:sg-ex}(a). The solution for a Spiral Galaxies puzzle may not be unique, but typically puzzles are designed to have a unique solution.

\begin{figure}[h]
\centering
\includegraphics[height=7em]{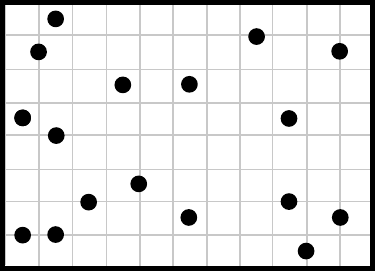}
\hfill
\includegraphics[height=7em]{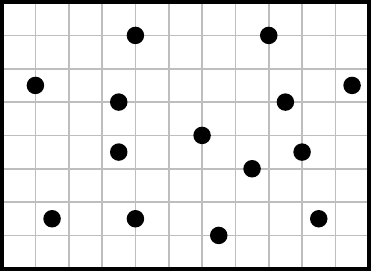}
\hfill
\includegraphics[height=7em]{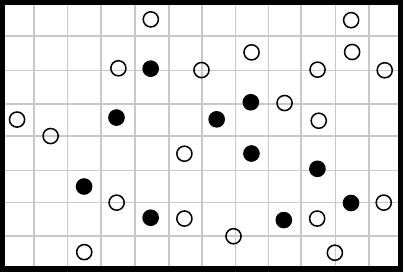}\\\bigskip
\includegraphics[height=7em]{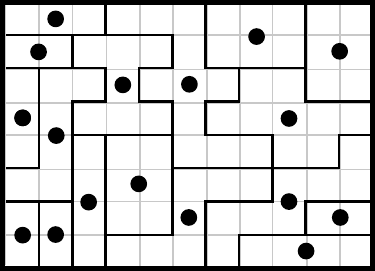}
\hfill
\includegraphics[height=7em]{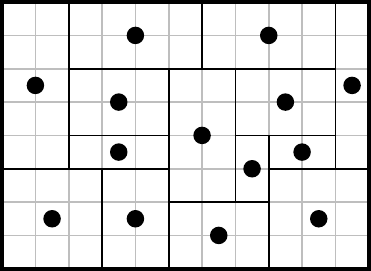}
\hfill
\includegraphics[height=7em]{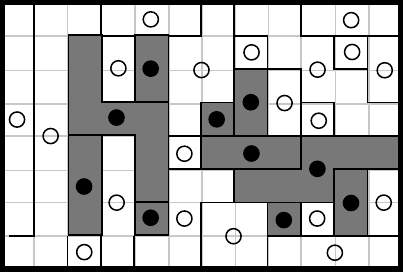}\\
\mbox{}\hfill(a)\hfill\mbox{}\hfill(b)\hfill\mbox{}\hfill(c)\hfill\mbox{}
  \caption
  { \small Several styles of Spiral Galaxies puzzles (top) and their solutions (bottom).
    (a) Classic Spiral Galaxies puzzle.
    (b) Rectangular Galaxies puzzle.
    (c) Spiral Galaxies puzzle with black and white centers such that the polyominoes containing the black centers in the solution yield a picture.
  }
  \label{fig:sg-ex}
\end{figure}

Friedman~\cite{f-sgpnp-02} proved Spiral Galaxies NP-complete for general polyomino galaxies.
His proof uses galaxies of area up to $22$.
Fertin et al.~\cite{fjk-tagsg-15} proved NP-completeness for galaxies of area at most $7$, while the puzzle becomes solvable in polynomial time if galaxies are restricted to have area at most $2$ or to be squares.
How much simpler can we make the galaxy shapes and still have Spiral Galaxies puzzles be NP-complete?

In this paper, we consider Spiral Galaxies with rectangular galaxies,
in two possible senses:
\begin{enumerate}
\item In \defn{Rectangular Galaxies}, the polyomino galaxies are
  required to be rectangles.  This is an additional restriction on the puzzle
  beyond $180^\circ$ rotational symmetry around the unique center
  contained in each rectangle.
\item Spiral Galaxies puzzles designed so that all their solutions
  (if any) have rectangular galaxies.  In other words, we have a
  \emph{promise} that all solutions are rectangular.
\end{enumerate}
Figure~\ref{fig:sg-ex}(b) shows an example of a Spiral Galaxies puzzle where
the unique solution is rectangular, so it is a valid input to either
Rectangular Galaxies or rectangular-promise Spiral Galaxies.

We prove that both of these puzzle variants are NP-complete and ASP-complete,
and that the corresponding problem of counting the number of solutions is
\#P-complete.
(Refer to~\cite[Chapter 28]{v-aa-10} and~\cite{y-ccfas-03,ys-ccfas-03} for definitions of \#P-completeness and ASP-completeness, respectively.)
In fact, hardness of rectangular-promise Spiral Galaxies implies
hardness of Rectangular Galaxies.
For Rectangular Galaxies,
we further show that the problem remains NP/ASP/\#P-complete
even when all solutions use only $1\times1$, $1\times3$, and $3\times1$
rectangles (a stronger promise),
and hence also if the puzzle is restricted to use only such rectangles.
This special case is in some sense the smallest interesting scenario:
if we restrict to only $1\times1$, $1\times2$, and $2\times1$ rectangles,
then each center determines the shape of its containing polyomino
(of area $1$ or $2$ depending on whether the center is at a cell center
or edge midpoint, respectively), making puzzles easy to solve.
Our proof for $1\times1$, $1\times3$, and $3\times1$ rectangles
also establishes NP-completeness of finding a noncrossing matching in ``distance-2 grid graphs'', whose vertices are a subset of the integer lattice and whose edges connect vertices at Euclidean distance exactly~$2$. 

In another variant of the puzzle, a subset of the centers is colored black, and the polyominoes that contain these centers reveals a picture once solved~\cite{gsp-18}; see Figure~\ref{fig:sg-ex}(c). Logic puzzles whose solutions yield pictures are a popular genre; one famous example is the nonogram~\cite{ueda96}. Such Spiral Galaxies puzzles are also the subject of a mathematical fonts~\cite{add-sgf-17}. Constructing an interesting puzzle such that its solution is a given target shape is nontrivial: while a valid puzzle trivially exists, by simply placing a center in every grid cell, the resulting puzzle is clearly not interesting.
We are thus also interested in finding the minimum number of centers such that there exist galaxies that exactly cover a given shape. We prove this puzzle design problem to be NP-complete. 



\subsection{Notation and Preliminaries}\label{sec:prob}

\input{notation}

\subsection{Organization}

The remainder of this article is organized as follows.
In Section~\ref{sec:sgr}, we prove that Spiral Galaxies is NP/ASP/\#P-complete when all solutions use rectangular galaxies.
In Section~\ref{sec:sgr13}, we show that Rectangular Galaxies is NP/ASP/\#P-complete when all solutions use rectangular galaxies of dimensions $1\times1$, $1\times3$, and $3\times1$.
In Section~\ref{sec:match}, we relate this problem to finding noncrossing perfect matchings in distance-2 grid graphs, proving NP/ASP/\#P-completeness of the latter.
In Section~\ref{sec:nrc}, we show that minimizing the number of centers when generating a puzzle with a given output shape is NP-hard.
In Section~\ref{sec:puzz}, we discuss the possibility of allowing multiple solutions in Spiral Galaxies puzzles, and we investigate a small puzzle related to a Spiral Galaxies font.
Finally, we conclude in Section~\ref {sec:conclusion}.

\section{Rectangular-Promise Spiral Galaxies}\label{sec:sgr}
\input{sgr}

\section{Rectangular Galaxies with $1\times1$, $1\times3$, and $3\times1$ Rectangles}\label{sec:sgr13}
\input{sgr13}

\section{Noncrossing Matching in Distance-2 Grid Graphs}\label{sec:match}
\input{matching}

\section{Minimizing Centers in Spiral Galaxies for a Given Shape}\label{sec:nrc}
\input{nrc}

\section{Multiple Solutions and the Font Puzzle}\label{sec:puzz}
\input{puzz}

\section{Conclusion and Discussion}\label{sec:conclusion}
We showed that both Rectangular Galaxies and rectangular-promise Spiral Galaxies are NP/ASP-complete even if the polyominoes are restricted to be rectangles of arbitrary size---and that counting the number of solutions is \#P-complete. Moreover, we proved that Rectangular Galaxies remains NP/ASP/\#P-complete even when only $1\times1$, $1\times3$, and $3\times1$ rectangles are possible. With the proof of the latter variant, we also showed that finding a noncrossing matching in distance-2 grid graphs is NP-complete. 
Finally, we proved NP-completeness of the design problem of finding the minimum number of centers such that there exist galaxies that exactly cover a given shape. The complexity of the latter problem when we aim for unique solutions remains open.

\section*{Acknowledgements}
We thank the anonymous referees for helpful comments.
C. S. was
partially supported by Jubileumsanslaget fr\r{a}n Knut och
Alice Wallenbergs Stiftelse.



\bibliography{lit}


\end{document}

%% file: notation.tex
In Spiral Galaxies, a \defn{board} consists of an $m\times n$ grid of unit squares called \defn{cells}.
A \defn{puzzle} consists of a board and a set of \defn{centers} placed at cell centers, edge midpoints, or grid vertices. Either all centers have the same color, or a subset of the centers may be colored black.
A solution to the puzzle consists of a partition of the grid into polyominoes, called \defn{galaxies}, such that each galaxy contains a single center and is $180^\circ$ rotationally symmetric about its center.

Our reductions are from the \textsc{Planar Positive 1-in-3 SAT} problem, a well-known NP-complete, ASP-complete, and \#P-complete problem~\cite{dyer1986planar, mulzer08, hunt1998complexity}:
\begin{definition}
  An instance of the \textsc{Planar Positive 1-in-3 SAT} problem is a \defn{formula} $F = (\mathcal{C},\mathcal{V})$ consisting of a set $\mathcal{C} = \{C_1, C_2, \dots, C_k\}$ of $k$ \defn{clauses} over $\ell$ \defn{variables} $\mathcal{V} = \{x_1, x_2, \dots, x_\ell\}$.
  Each clause in $F$ is a set of at most three variables, and in particular contains variables only in their positive form (no negation).  
  The variable--clause incidence graph $G$ is planar. 
  A clause is satisfied if and only if it contains exactly one \textsc{true} variable, and the formula $F$ is satisfied if and only if all its clauses are satisfied.
  The goal is to find a Boolean assignment to the variables that satisfies the formula~$F$.
\end{definition}

%% file: sgr.tex
In this section, we show that solving Spiral Galaxies puzzles promised to have solutions with only rectangular galaxies is NP-complete and ASP-complete, and that the corresponding problem of counting the number of solutions is \#P-complete. 

We give a reduction from \textsc{Planar Positive 1-in-3 SAT}. Given an instance $F$ of \textsc{Planar Positive 1-in-3 SAT} with incidence graph $G$, we show how to turn a rectilinear planar embedding of $G$ into a Spiral Galaxies puzzle $P$ such that a solution to $P$ yields a solution to $F$, thereby showing NP-completeness.
Furthermore, there will be a one-to-one correspondence between solutions of $P$ and solutions of $F$, showing \#P-completeness and ASP-completeness. 

\textbf{Overview and filler gadget.}
At a high level, our reduction consists of several gadgets:
``variable'' gadgets representing the variables of~$F$;
``wire'' gadgets to connect variables to clauses;
and ``clause'' gadgets to form the clauses of~$F$.
Refer ahead to Figure~\ref{fig:global} for a complete example of the reduction.

 \begin{figure}
\centering
\hspace*{.25\textwidth}
\comic{.1\textwidth}{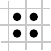}{(a)}\hfill
\comic{.1\textwidth}{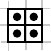}{(b)}
\hspace*{.25\textwidth}
  \caption{\small (a) Filler gadget to fill the space in between all other gadgets. (b) Forced solution.}
  \label{fig:rect-face}
\end{figure}

For each region of the board that is not part of these gadgets,
we fill the region with a \defn{filler gadget},
which has a center in every cell of the region.
Figure~\ref{fig:rect-face} shows the filler gadget for a $2 \times 2$ region.
More generally, in every $2 \times 2$ square within the region,
we can argue locally that the four corresponding galaxies
must each consist of a single cell (the one containing the center):
the edges between cells must be galaxy boundaries
to separate the centers into separate galaxies,
and then $180^\circ$ rotational symmetry forces galaxy boundaries
around the $2 \times 2$ square.
As long as the region is the union of such $2 \times 2$ squares,
the filler gadget must consist entirely of single-cell galaxies,
without any interaction with the other gadgets in the construction.
On the other hand, if the region has a width-1 row or a
height-1 column (``thickness~1''), then the galaxy at the center of that cell
might includes cells surrounding the filler gadget.
We must therefore guarantee that every region between other gadgets
is the union of $2 \times 2$ squares (``thickness~2''),
so that the filler gadget has a forced solution of single-cell galaxies.

\begin {figure}
  \centering
  \includegraphics {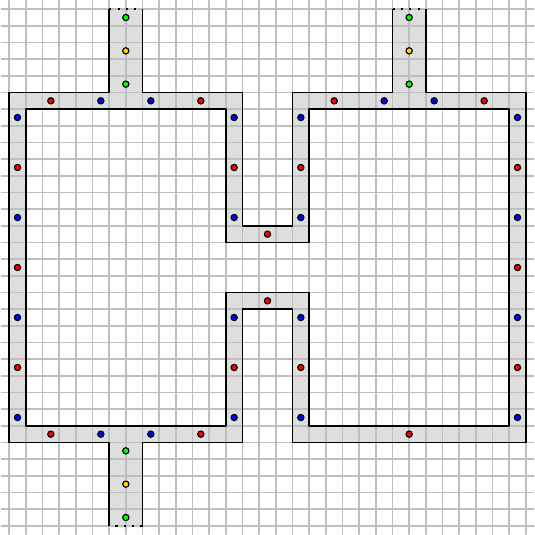}
  \caption 
  {Overall construction of a variable gadget as a sequence of bumps on the top and bottom, where each bump can have a single connection to an incident wire on that side (Figure~\ref{fig:variable-wire}). Bumps without such a connection, such as the one in the bottom right, use a single center.
  }
  \label{fig:variable-loop-fix}
\end {figure}

\textbf{Variable gadget and loop.}
Figure~\ref{fig:variable-loop-fix} shows the overall plan for a variable
gadget: a thickness-$1$ ``variable loop'' that follows a long horizontal
rectangle with regularly spaced bumps on the top and bottom sides,
where each bump has zero or one connection to an incident wire
(which has thickness~$2$).

In more detail, a \defn{variable loop} is a thickness-$1$ loop
built out of the subgadgets in Figures~\ref{fig:variable-straight}
and~\ref{fig:variable-corner}.
Every center is at a cell center, spaced modulo $3$
along the thickness-$1$ loop.
Each center in the middle of a straight piece can choose an $x \times 1$
rectangular galaxy for $x \in \{1,3,5\}$,
which then forces the next galaxy along the straight part to be $6-x \times 1$,
etc., as in Figure~\ref{fig:variable-straight}(b--d).
Each corner subgadget of Figure~\ref{fig:variable-corner}
forbids the center at distance $2$ from the corner
from having a $1 \times 1$ galaxy,
as then the galaxy centered at distance $1$ from the corner
fails to be $180^\circ$ rotationally symmetric;
see Figure~\ref{fig:variable-corner}(d).
Our ``bumpy rectangle'' design from Figure~\ref{fig:variable-loop-fix}
guarantees that every straight portion of a variable loop
is adjacent to at least one corner, and
we can further arrange that all corners have the same alignment.
These properties forbid one global pattern
(subfigure~(d) with blue $5 \times 1$ galaxies)
and leave two possible solution patterns:
all galaxies are $3 \times 1$ or $1 \times 3$ (subfigure~(b)),
and galaxies alternate between blue $1 \times 1$
and red $5 \times 1$/$1 \times 5$ (subfigure~(c)).
These two solutions to the variable gadget correspond to
setting the variable \textsc{false} or \textsc{true}, respectively.

\begin {figure}
  \centering
  \includegraphics {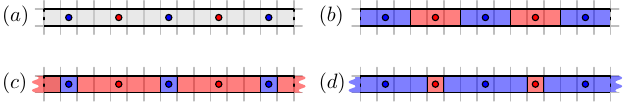}
  \caption 
  { (a) A straight piece of a variable gadget.
    (b--c) Two intended valid states.
    (d) Undesired but valid state.
  }
  \label{fig:variable-straight}
\medskip
  \centering
  \includegraphics {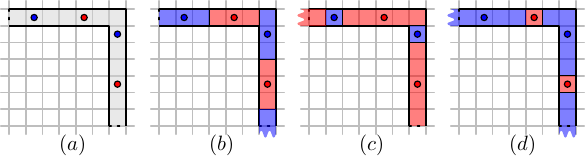}
  \caption 
  { (a) A corner of a variable gadget.
    (b--c) Two valid states.
    (d) The third state is no longer valid.
  }
  \label{fig:variable-corner}
\end {figure}

\textbf{Wire gadget.}
While thickness $1$ was useful to force the variable gadget to have exactly
two solutions, it seems difficult to copy that truth value onto multiple
thickness-$1$ wires.
We thus introduce thickness-$2$ \defn{wire gadgets}
as shown in Figure~\ref{fig:wire}.
In their simplest form, a wire gadget is a vertical width-$2$ rectangle
with a center at edge midpoints in the even rows modulo~$2$.
The two intended solutions use alternating $2 \times 1$ and $2 \times 3$
galaxies as shown in Figure~\ref{fig:wire}(b--c),
corresponding to \textsc{false} and \textsc{true} signals respectively,

Unfortunately, the basic wire also allows two nonrectangular solutions,
as shown in Figure~\ref{fig:wire}(d--e).
Each of these undesired solutions is prevented by a corresponding
\defn{shift gadget} shown in Figures~\ref{fig:shift-right}
and~\ref{fig:shift-left} respectively.
In either gadget, covering the corner pixels in the middle row
forces one side of the wire to use rectangles.
By guaranteeing that each wire shifts at least once to the right
and at least once to the left (including possibly one shift
canceling out the other), we avoid the undesired solutions,
forcing one of the two intended solutions in Figure~\ref{fig:wire}(b--c).


\begin {figure}
  \centering
  \includegraphics {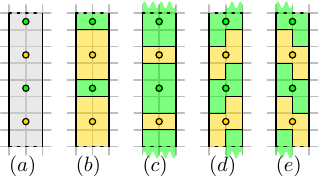}
  \caption 
  { (a) A wire gadget for connecting a variable to a clause.
    (b--c) Intended valid rectangular states.
    (d--e) Undesired but valid states.
  }
  \label{fig:wire}
\end {figure}
\medskip
\begin {figure}
  \centering
  \includegraphics {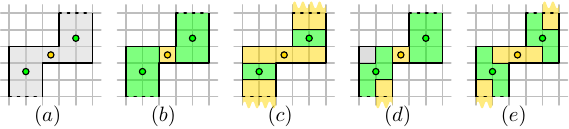}
  \caption 
  { (a) A shift to the right.
    (b---c, e) Three states are still valid.
    (d) The undesired wire state from Figure~\ref{fig:wire}(d) is no longer valid.
  }
  \label{fig:shift-right}
\medskip
  \centering
  \includegraphics {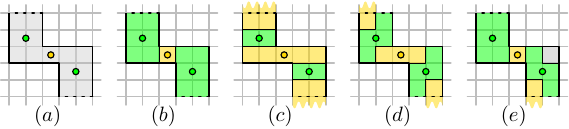}
  \caption 
  { (a) A shift to the left.
    (b---d) Three states are still valid.
    (e) The undesired wire state from Figure~\ref{fig:wire}(e) is no longer valid.
  }
  \label{fig:shift-left}
\end {figure}

Each shift gadget can shift the wire horizontally
by an arbitrary amount $\geq 3$,
as shown in Figure~\ref{fig:shift-var}.
The height-1 transition row makes it easy to argue that the same
solutions in Figures~\ref{fig:shift-right} and~\ref{fig:shift-left}
are forced.
Thus a wire can approximately follow any desired $y$-monotone path
that starts and ends with a vertical segment.
The $x$ coordinates of the start and end vertical segments
can be specified exactly;
the only approximations are that the shifts need $\Omega(1)$ room to navigate,
and we always use at least two shifts.

\begin {figure}
  \centering
  \includegraphics {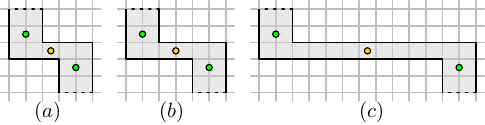}
  \caption 
  { Shift gadgets from Figures~\ref{fig:shift-right} and~\ref{fig:shift-left}
    can shift by an arbitrary horizontal amount $\geq 3$.
    (a, c) An even shift places the middle center at a cell center.
    (b) An odd shift places the middle center at an edge midpoint.
  }
  \label{fig:shift-var}
\end {figure}

\textbf{Connecting variables to wires.}
Figure~\ref{fig:variable-wire} shows how to communicate the value of a
variable loop into the signal of a wire.
The vertical thickness-$2$ wire attaches to two consecutive empty cells
of a horizontal segment of the variable loop.
If the wire's galaxies capture those two cells, then the two adjacent
centers of the variable loop must have $1 \times 1$ galaxies,
forcing alternation between $1 \times 1$ and $1 \times 5$ thereafter,
as in Figure~\ref{fig:variable-wire}(c).
Conversely, if the wire's galaxies do not capture those two cells,
then the two adjacent centers of the variable loop must capture them,
forcing the one other solution of the variable loop with
$1 \times 3$ rectangles, as in Figure~\ref{fig:variable-wire}(b).

\begin {figure}
  \centering
  \includegraphics {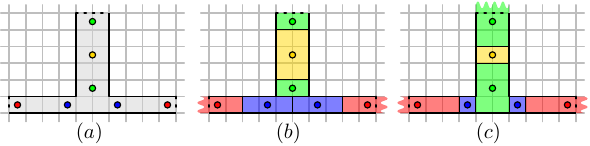}
  \caption 
  { (a) A place where a wire from the variable loop is connected.
    (b--c) Two compatible valid states.
  }
  \label{fig:variable-wire}
\end {figure}

Figure~\ref{fig:wires-need-U-turns} shows that the unintended solutions of
wire gadgets from Figure~\ref{fig:wire}(d--e) can propagate into the
variable gadget.  Thus connections to variable gadgets are insufficient to
fix these problems, which is why we needed to introduce the shift gadgets of
Figure~\ref{fig:shift-right} and~\ref{fig:shift-left}
to fix wire gadgets locally.

\begin {figure}
  \centering
  \includegraphics {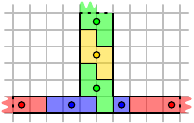}
  \caption 
  { Why we needed U-turns in wire gadgets: connections to variables
    are insufficient to force the intended rectangular solutions.
  }
  \label{fig:wires-need-U-turns}
\end {figure}

If we connected multiple wires to the same straight segment
of a variable gadget, then the straight segment between the two connections
would not be incident to any variable corners,
so it could be in the third state of Figure~\ref{fig:variable-straight}(d).
This is why we limit each bump of the variable loop to
having at most one connection to a wire,
so that every connection gadget from Figure~\ref{fig:variable-wire}
is surrounded on either side
by a corner gadget from Figure~\ref{fig:variable-corner}.

Each connection between variable and wire
removes one alternation from the variable gadget,
as indicated by the coloring in Figure~\ref{fig:variable-wire}.
If the total number of connections from a variable gadget is odd,
then we would have a parity mismatch around the variable loop.
Figure~\ref{fig:variable-loop-fix} illustrates a parity fix
by placing a single center in the middle of a bump of the variable gadget
that has no connection to a wire (similar to Figure~\ref{fig:shift-var}).
For simplicity, we can use this construction for every bump
with no connection to a wire.
Because the two adjacent corners turn in the same direction,
this center can be covered in exactly two ways:
with a rectangle that is either the full horizontal width
or $2$ smaller.

%
%
%
%
%

\begin {figure}
  \centering
  \includegraphics {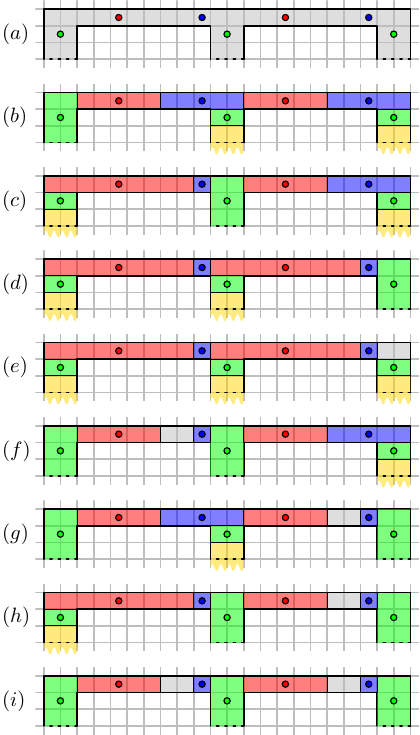}
  \caption 
  { The clause gadget connects three wire gadgets (green).
    (b--d) Valid solutions with exactly one true signal.
    (e) Impossibility with zero true signals.
    (f--h) Impossibility with two true signals.
    (i) Impossibility with three true signals.
  }
  \label{fig:clause}
\end {figure}

\textbf{Clause gadget.}
The last gadget is the clause gadget shown in Figure~\ref{fig:clause},
which implements a 1-in-3 constraint between three incident wire gadgets
via a $1 \times 22$ rectangle.
Although we draw all three wire gadgets as below the clause gadget,
each wire could in fact be either above or below the clause gadget.
Figure~\ref{fig:clause}(b--d) illustrates that the clause gadget has a
(unique) solution when exactly one of the wires carries a true signal.
Indeed, any true signal represented by a wire with a galaxy covering
two cells of the clause gadget locally forces the galaxies of the
one or two adjacent centers within the clause, which then propagates to force
the galaxies of the other centers within the clause.
This forced solution turns out to be inconsistent with
any of the other wires carrying a true signal, causing overlapping galaxies
or uncovered cells as shown in Figure~\ref{fig:clause}(e--i).
Therefore the clause gadget has a (unique) solution if and only if
exactly one of the incident wires carries a true signal.

\textbf{Putting pieces together.}
For the global construction, we construct an orthogonal planar embedding of $G$
where each vertex is represented by a horizontal line segment
with endpoints on the grid,
and each edge is represented by a polygonal line on the grid
that is $y$-monotone and starts and ends with a vertical segment.
Such an embedding can be obtained by
constructing a planar straight-line embedding of $G$
with vertices on an $O(n) \times O(n)$ grid \cite{schnyder90};
applying the transformation of \cite[Theorem~5]{biedl14}
into a $y$-monotone orthogonal drawing where vertices are horizontal segments;
and then staggering vertices with equal $y$ coordinates.
The resulting grid size is $O(n) \times O(n^2)$.

Figure~\ref{fig:global}(a) shows a particularly structured planar embedding
of $G$
where the variables all lie on a horizontal line,
every clause is either above or below this line,
and variable--clause connections do not cross the line or each other.
Then it is easy to draw each variable--clause connection with a $y$-monotone
path that starts and ends with a vertical segment.
Mulzer and Rote \cite{mulzer08} proved that \textsc{Planar Positive 1-in-3 SAT}
is NP-complete in this form.  Their reduction has also been observed to be
the end of a chain of parsimonious reductions from 3SAT
\cite{6.890-planar-1-in-3},
so it establishes ASP- and \#P-completeness too.
We describe the more general reduction in the previous paragraph
so as to not rely on this unpublished observation,
but use the simpler form for figures.

\begin {figure}
  \centering
  \includegraphics[width=\linewidth]{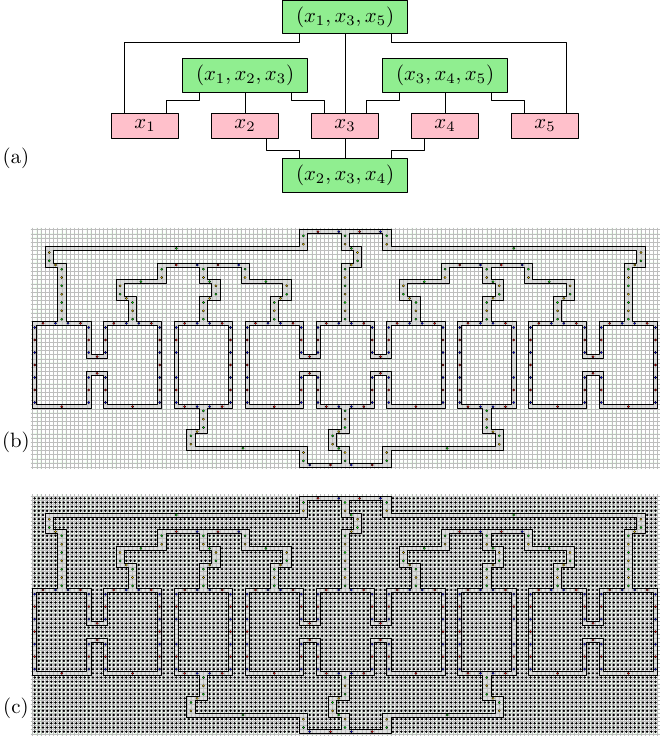}
  \caption 
  { (a) An embedded instance of Planar Positive 1-in-3 SAT, with all variables on a common horizontal line and all clauses either above or below the line.
    (b) The corresponding gadgets of Spiral Galaxies, excluding filler.
    (c) Full Spiral Galaxies puzzle with filler gadgets.
  }
  \label{fig:global}
\end {figure}

As illustrated in Figure~\ref{fig:global}(b),
we scale up the orthogonal planar embedding by a constant factor;
replace each variable or clause vertex with a variable gadget or clause gadget
respectively,
and replace each $y$-monotone edge with a wire gadget
(including at least one right shift and at least one left shift).
Finally, we place the filler gadget in the regions between other gadgets,
as shown in Figure~\ref{fig:global}(c).
Because each gadget has a unique solution corresponding to a Boolean
assignment, this reduction is parsimonious.

Finally we analyze alignment issues in this construction.
The shift gadgets of Figure~\ref{fig:shift-var} provide exact
horizontal positioning, so there are no constraints on horizontal
positioning of gadgets.
However, vertical positioning controls the parity of the \textsc{false}
and \textsc{true} values of a wire, as in the yellow and green coloring
of Figure~\ref{fig:wire}.
We place all variable loops to have the same vertical positioning modulo~$4$.
Thus we obtain consistent heights modulo $4$
for the row of a wire between a \textsc{false} (yellow) and \textsc{true}
(yellow) center that needs to be aligned with a clause:
for wires extending upward, this is the row above a \textsc{true} center, while
for wires extending downward, this is the row below a \textsc{true} center.
Observe in Figure~\ref{fig:variable-loop-fix} that these rows have the
same height modulo $4$ within a single variable loop,
and that wire gadgets of Figure~\ref{fig:wire} and shift gadgets of
Figures~\ref{fig:shift-right}--\ref{fig:shift-var}
preserve this height modulo~$4$.
We can then place all clause gadgets at such $y$ coordinates.

A solution to an $m\times n$ Spiral Galaxies puzzle can be verified in polynomial time. We conclude: 

\begin{theorem}
	Solving a Spiral Galaxies puzzle whose solutions have only rectangular galaxies is NP-complete and ASP-complete, and counting the number of solutions is \#P-complete.
\end{theorem}

%% file: sgr13.tex
In this section, we show that solving 
Rectangular Galaxies puzzles is NP-complete and ASP-complete, even with the promise that every solution uses only $1\times 1$, $1\times 3$, and $3\times 1$ galaxies, and that the corresponding problem of counting the number of solutions is \#P-complete. 
Again we give a reduction from \textsc{Planar Positive 1-in-3 SAT}.
Given an instance $F$ of \textsc{Planar Positive 1-in-3 SAT} with incidence graph $G$, we show how to turn a rectilinear planar embedding of $G$ into a Rectangular Galaxies puzzle $P$ such that a solution to $P$ yields a solution to $F$, thereby showing NP-completeness. Furthermore, there will be a one-to-one correspondence between solutions of $P$ and solutions of $F$, showing \#P-completeness and ASP-completeness. 

In this section, all galaxies will be $1\times 1$, $1\times 3$, or $3\times 1$.
To ease the description, our figures draw these galaxies as edges of length 2 ($1\times 3$ and $3\times 1$) or nonexisting edges ($1\times 1$).
Instead of drawing the galaxy centers (dots), we draw small boxes (squares) that reflect the endpoints of these edges.
Boxes at distance 2 can be connected by a rectilinear edge, and centers will be located at the middle of each potential edge; see Figure~\ref{fig:3-1-disks-and-centers}.
Hence, any $1\times 3$ or $3\times 1$ galaxy will cover both boxes, denoted by an edge between these two boxes; any $1\times 1$ galaxy will not extend over the boxes, and is shown by a nonexisting edge between the boxes.

\begin{figure}
\centering
\comicII{.226\textwidth}{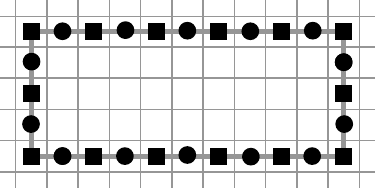}
  \caption{\small We place centers (dots) on the middle point of potential edges (light gray) between any pair of boxes (squares). In the remainder of this section, we show only the boxes, not the centers (dots). }
  \label{fig:3-1-disks-and-centers}
\end{figure}


\textbf{Overview and filler gadget.}
At a high level, our reduction consists of two main gadgets:
``variable'' gadgets representing the variables of~$F$;
and ``clause'' gadgets to form the clauses of~$F$.

As in Section~\ref{sec:sgr}, we fill any region not covered by these gadgets
with a \defn{filler gadget}, which places a center at every cell of the region.
Figure~\ref{fig:rect-face} shows the filler gadget for a $2 \times 2$ region.
The centers of a filler gadget must be contained in $1 \times 1$ galaxies,
provided every uncovered region is the union of $2 \times 2$ squares.

\textbf{Variable gadget.}
Figure~\ref{fig:3-1-var} illustrates the \defn{variable gadget},
which is an alternating pattern of boxes.
While the figure shows the gadget in a simple rectangular form,
variables can also turn arbitrarily, as long as each box has
distance $2$ to exactly two other boxes (i.e., the variable does not
get too close to itself).

\begin{lemma}\label{var-1-3-only-2-sols}
Each variable gadget has exactly two possible solutions (shown in Figure~\ref{fig:3-1-var}(b) and (c)).
\end{lemma}
\begin{proof}
Because of the filler gadgets, no galaxy centered at a center
of the variable gadget can extend beyond the grid cells
covered by the potential edges of the variable gadget.
Hence, every rectangular galaxy must have a width or height of 1,
which we label as ``vertical'' or ``horizontal'' respectively.
Because the centers are placed with distance 2 and galaxies must be
$180^\circ$ rotationally symmetric about their center,
horizontal and vertical galaxies have a width and height respectively
in $\{1,3\}$.
Because galaxies may not overlap (we aim to decompose the grid),
$1\times 3$/$3\times 1$ galaxies must alternate with $1 \times 1$ galaxies.
There exist exactly two possible solutions that fulfill these conditions.
\end{proof}

Each of the two possible solutions corresponds to one truth assignment
for the variable, \textsc{true} and \textsc{false}.
We route each variable gadget to interact with every clause that contains
the variable.  Equivalently, we can imagine each variable gadget as having
a ``tentacle'' to visit every incident clause.

\begin{figure}
\centering
\hspace*{.1\textwidth}
\comic{.226\textwidth}{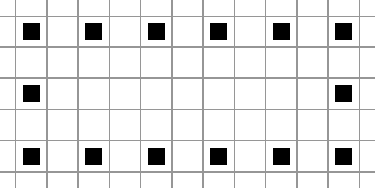}{(a)}\hfill
\comic{.226\textwidth}{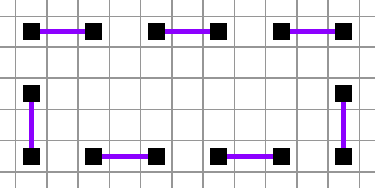}{(b)}\hfill
\comic{.226\textwidth}{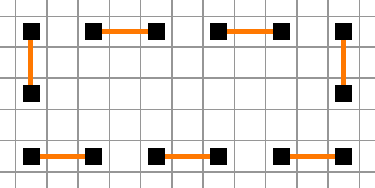}{(c)}
\hspace*{.1\textwidth}
  \caption{\small (a) Variable gadget with two possible states (b) and (c) corresponding to a truth assignment of \textsc{true} and \textsc{false} of the corresponding variable. }
  \label{fig:3-1-var}
\end{figure}

\begin{figure}
\centering
\comic{.638\textwidth}{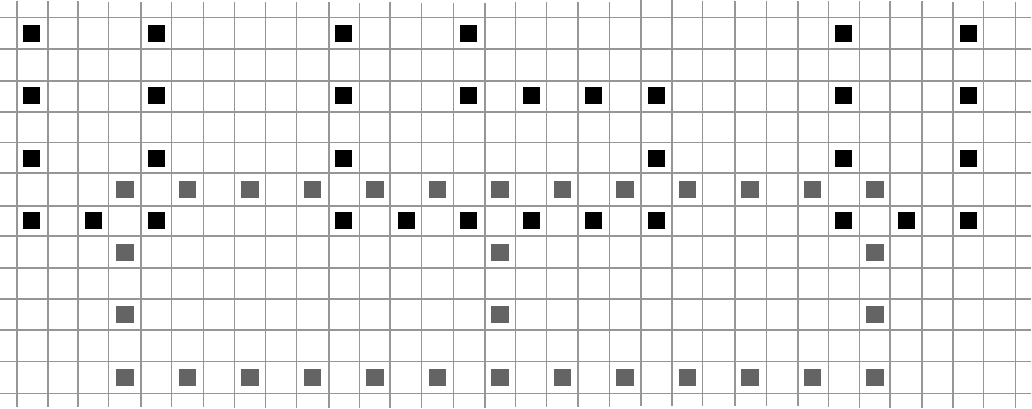}{(a)}\hfill
\comic{.638\textwidth}{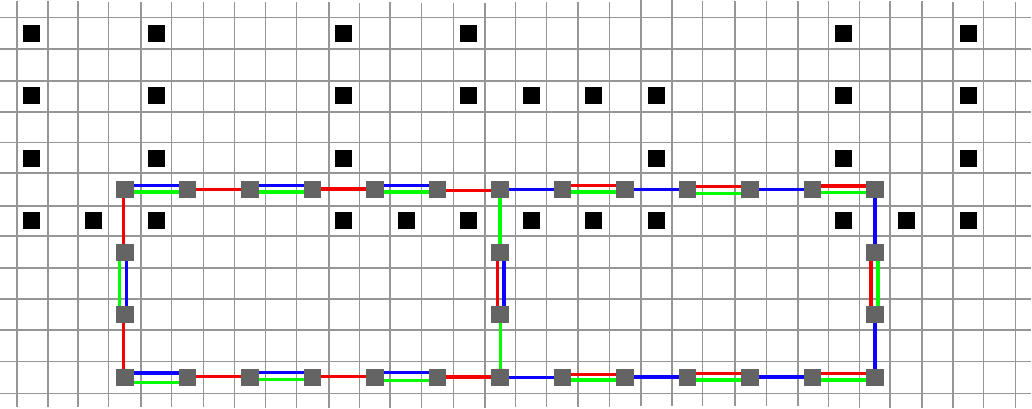}{(b)}\\
\comic{.638\textwidth}{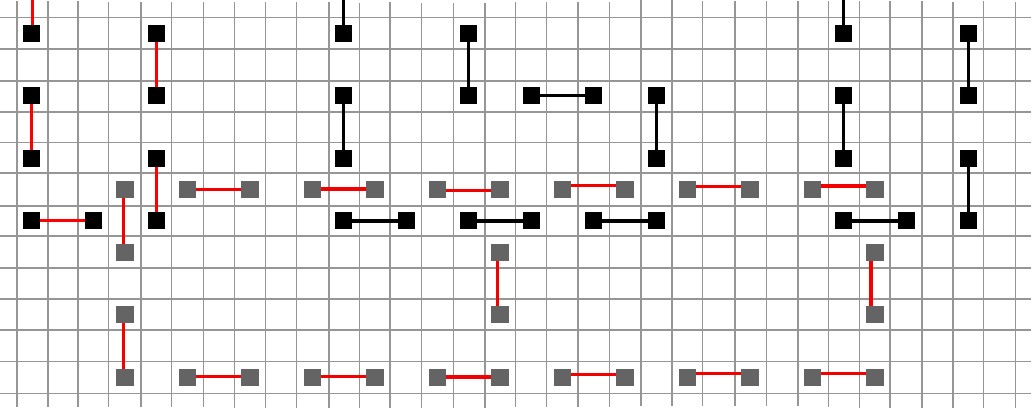}{(c)}\hfill
\comic{.638\textwidth}{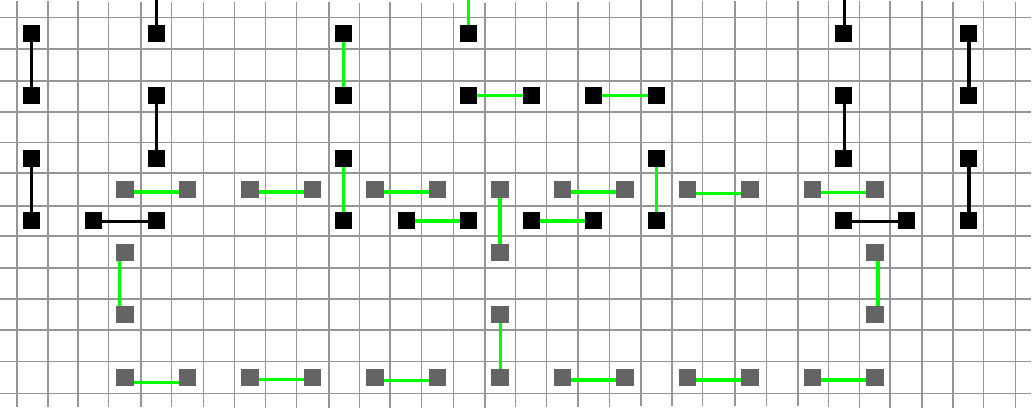}{(d)}\\
\comic{.638\textwidth}{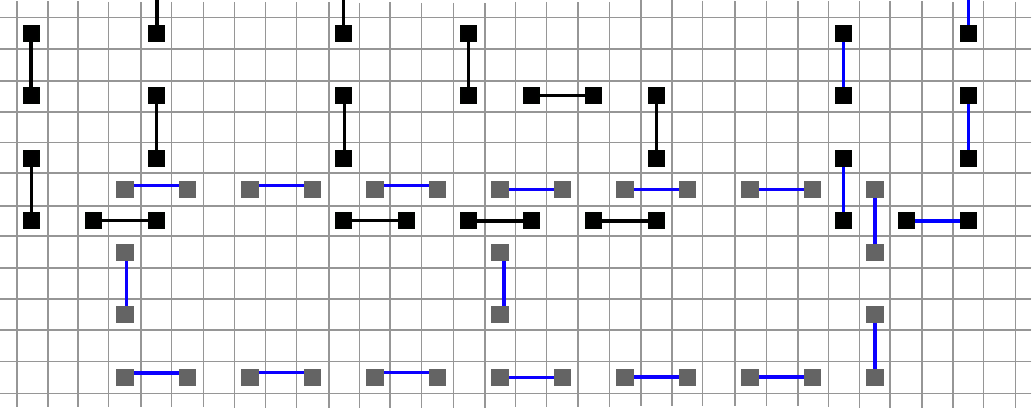}{(e)}\vspace*{-.2cm}
  \caption{\small (a) Clause gadget in gray, with the three incoming variable gadgets in black.  (b) The three possible states of the clause gadget in red, green, and blue. (c)--(e) Each state of the clause gadget with the corresponding assignments of the variable gadgets; the (only) true variable is shown in the same color as the clause edges.}
  \label{fig:3-1-clause}
\end{figure}

\textbf{Clause gadget.}
Figure~\ref{fig:3-1-clause} shows the clause gadget,
where the clause gadget itself is drawn in gray,
and the three incident variable gadgets are drawn in black.
There are three possible states of the clause gadget,
drawn in red, green, and blue in Figure~\ref{fig:3-1-clause}(b).
Each state forces exactly one of the variables' truth assignments
to be \textsc{true} and all the other variables to be \textsc{false},
as shown in Figure~\ref{fig:3-1-clause}(c)--(e).

\textbf{Putting pieces together.}
For the global construction, we start from an planar embedding of $G$
with edges routed orthogonally on an $O(n) \times O(n)$ grid
\cite{biedl98}, scaled up by a constant factor.
Then we locally replace each clause by a single clause gadget,
and replace each variable by a sufficiently large variable gadget
that extends (via ``tentacles'') to all clauses in which it appears.
All variable gadgets use one parity
(e.g., all boxes have even $x$ and even $y$ coordinates),
while all clause gadgets 
use the other parity (e.g., all boxes have odd $x$ and odd $y$ coordinates).
Hence, we can always place all gadgets and connect them together as desired.
Because each gadget has a unique solution corresponding to a Boolean
assignment, this reduction is parsimonious.

A solution to an $m\times n$ Rectangular Galaxies puzzle can 
be verified in polynomial time. We conclude:

\begin{theorem}\label{th:sgr13}
	Solving a Rectangular Galaxies puzzle whose solutions have only $1\times 1$, $1\times 3$, and $3\times 1$ galaxies is NP-complete and ASP-complete, and counting the number of solutions is \#P-complete.
\end{theorem}

%% file: matching.tex
In this section, we prove an equivalence between Spiral Galaxies with
$1 \times 1$, $1 \times 3$, and $3 \times 1$ galaxies and
``noncrossing perfect matchings'' in ``distance-2 grid graphs''.
First we define the quoted terms.
A \defn{distance-2 grid graph} $G = (V, E)$ has vertices $V \subset \Z \times \Z$ at a finite subset of the integer lattice, and edges $E = \{(v,w) \mid v, w \in V, \|v-w\| = 2\}$ between all pairs of vertices at Euclidean distance exactly $2$.
A \defn{noncrossing perfect matching} in a distance-2 grid graph $G = (V, E)$
is a partition of the vertices $V$ into $|V|/2$ pairs
such that the length-$2$ segments connecting paired vertices
do not intersect (including at endpoints).

Next we observe that distance-2 grid graphs can always be decomposed into two independent components.
Call a vertex $v = (x,y) \in V$ \defn{even} if $x+y$ is even,
and \defn{odd} if $x+y$ is odd.
Because edges connect vertices of distance exactly $2$,
which is even, vertices of different parity cannot be connected together.
Thus even/odd forms a partition of any distance-2 grid graph.
In the remainder, we restrict to \defn{even} distance-2 grid graphs
whose vertices are all even.

Call a Spiral Galaxies puzzle \defn{even} if all centers are at cell centers and all \emph{empty} cells are at even positions.
We claim that there is a one-to-one correspondence between noncrossing perfect matchings in even distance-2 grid graphs and solutions to even Spiral Galaxies puzzles restricted to $1\times 1$, $1\times 3$, and $3\times 1$ galaxies:

\begin {lemma} \label {lem:equiv}
  Given an even distance-2 grid graph, we can efficiently construct an even Spiral Galaxies puzzle restricted to $1\times 1$, $1\times 3$, and $3\times 1$ galaxies and correspond one-to-one with noncrossing perfect matchings in the distance-2 grid graph. Vice versa, we can reduce from a restricted even Spiral Galaxies puzzle to noncrossing perfect matching in an even distance-2 grid graph.
\end {lemma}

\begin {proof}
  Given an even distance-2 grid graph $G = (V, E)$,
  we construct an even Spiral Galaxies puzzle as follows.
  The board is the bounding box of the vertices $V$ enlarged by $1\over2$ on all sides, aligned so that every vertex of $V$ is placed at the center of a board cell.
  For every \emph{even} cell $(x, y)$ of the board, we place a galaxy center at its center if and only if $(x, y) \notin V$; that is, the vertices of $G$ correspond to empty cells of the board.
  For every \emph{odd} cell $(x, y)$ of the board, we always place a galaxy center at its center.
  
  Similarly, from a given even Spiral Galaxies puzzle, we can construct an even distance-2 grid graph by creating a vertex for every empty board cell. (The edges are implied by the vertices by the definition of distance-2 grid graph.)
  
  We claim that noncrossing perfect matchings of $G$ correspond one-to-one with solutions to the Spiral Galaxies instance restricted to only $1\times1$, $1\times3$, and $3\times1$ galaxies.
  Indeed, all empty cells in the puzzle must be captured by a galaxy, and can be captured only in pairs on opposite sides of a galaxy center (in an odd cell), which in $G$ corresponds to an edge connecting two vertices.
\end {proof}

Combining Lemma~\ref {lem:equiv} and Theorem~\ref{th:sgr13}
(or equivalently, interpreting the boxes in figures as graph vertices),
we obtain hardness for noncrossing perfect matching:
\begin{corollary}
Noncrossing perfect matching in distance-2 grid graphs
is NP-complete and ASP-complete, and counting the number of solutions is \#P-complete.
\end{corollary}

%% file: nrc.tex
In this section, we are given a shape $\mathcal{S}$ on a Spiral Galaxies board, and we aim to find the minimum number of centers such that there exist galaxies with these centers that exactly cover the given shape $\mathcal{S}$. We show that this problem is NP-complete by another reduction from \textsc{Planar Positive 1-in-3 SAT}.

\begin{theorem}\label{th:nrc}
	It is NP-complete to minimize the number of centers on a Spiral Galaxies board such that galaxies with these centers exactly cover a given shape~$\mathcal{S}$.
\end{theorem}

%
For our reduction, the input is a \textsc{Planar Positive 1-in-3 SAT} formula, and the output is a ``picture'' --- a subset of cells that should be ``colored black''. In the remainder of the section, we focus on the cells that are part of the picture; we will say that the cells which are not part of the picture are \defn{forbidden}.

First we will define some \defn{low-level gadgets}, which each have a constant size.
Then we will build \defn{high-level gadgets}, which consist of multiple low-level gadgets, and which will be used to encode the 1-in-3-SAT formula.

\subsection {Low-level gadgets}

\textbf{Local center gadgets.}
First we introduce \defn{local center gadgets}: thin constructions with unique shapes, which ensure that there must be at least one center in each of them; see Figure~\ref {fig:center-gadget}. 
The idea will now be to construct a shape that can be covered with exactly this set of centers if and only if the 1-in-3-SAT instance is satisfiable. 

\begin{figure}
\centering
\includegraphics [scale=0.75]{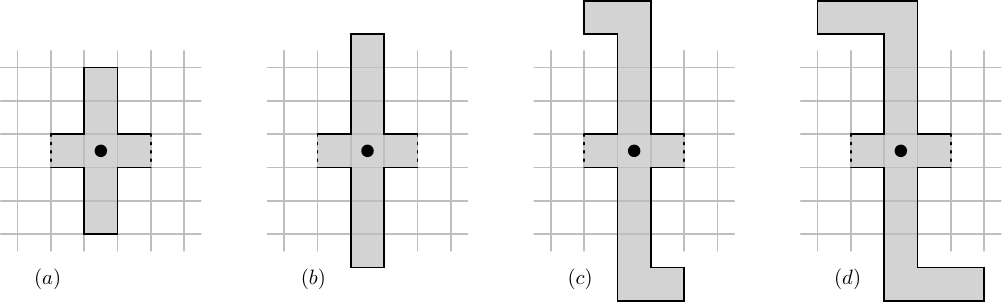}
  \caption{\small The local center gadget must have at least one center. (a)--(d) Different shapes ensure we cannot include them in larger galaxies.}
  \label{fig:center-gadget}
\end{figure}

Specifically, each local center gadget consists of a single cell with a \defn{desired galaxy center}. Its four direct neighbors are all empty, and its four diagonal neighbors are all forbidden. The cells left and right of the center will serve as connections to other low-level gadgets (we describe the gadgets as if they are vertically oriented; in the actual construction they may be rotated by $90^\circ$). The cells above and below the center are extended by symmetrically shaped paths of empty cells which are surrounded by forbidden cells.

We can immediately observe that, if there is \emph{not} a center at the location of the desired galaxy center in a local center gadget, then either the entire gadget will be part of a larger galaxy with a center outside the gadget, or there must be at least two centers placed inside the gadget. In the first case, because of the symmetry of the galaxies, there must be a second copy of the local center gadget that is also part of the same galaxy. For this reason, we use not one but several different versions of the gadget.
As we will see later, $O(1)$ different version will be sufficient.

 \begin{figure}
\centering
\includegraphics [scale=0.75]{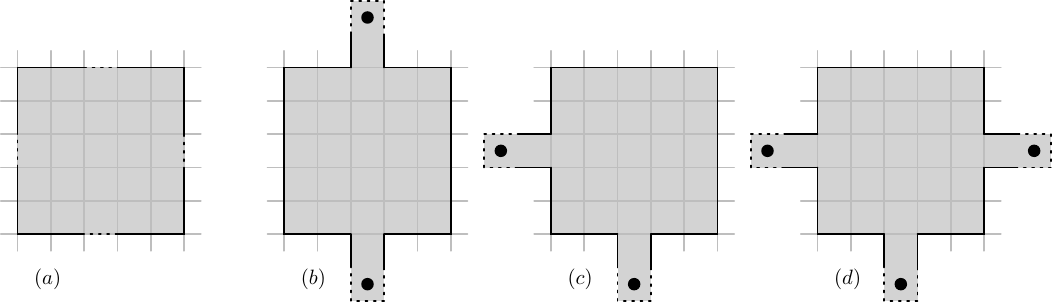}
  \caption{\small (a) The block gadget. (b) A straight block. (c) A corner block. (d) A clause block.}
  \label{fig:block-gadget}
\end{figure}

\textbf{Block gadgets.}
The second low-level gadget we use is the \defn{block gadget}.
A block gadget is simply a $5\times5$ room surrounded by forbidden cells, except for possibly in the centers of its four sides. 
The idea is that a block gadget will not contain galaxy centers in an optimal solution, but rather will be connected to local center gadgets and be part of their galaxies. 
Depending on which of the sides of a block gadget are connected to other gadgets, we distinguish \defn{straight blocks}, \defn{corner blocks}, and \defn{clause blocks}; see Figure~\ref {fig:block-gadget}.

\begin {observation} \label {obs:desire}
  Suppose we place local center gadgets and block gadgets in such a way that, for each block gadget $B$, all local center gadgets that $B$ is connected to have different shapes.
  Then we must place at least one galaxy center inside each local center gadget.
\end {observation}

\begin{figure}
\centering
\includegraphics [scale=0.75]{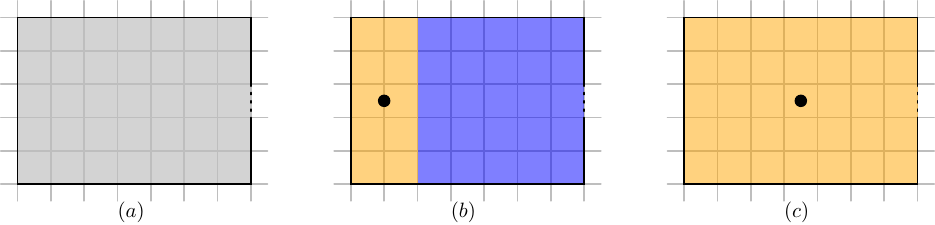}
  \caption{\small An end gadget has one center, whether or not a block is used by the galaxy next to it.}
  \label{fig:end-gadget}
\end{figure}

\textbf{End gadgets.}
Finally, our last low-level gadget is the \defn{end gadget}. 
An end gadget is similar to a bock gadget, but it is a $7\times5$ room and it has only one connection, at the center of one of its sides of length $5$; see Figure~\ref {fig:end-gadget}.

The idea is that, if we place an end gadget at the end of a chain that satisfies the conditions of Observation~\ref {obs:desire}, then each end gadget will require at least one galaxy center to be placed inside it, because the center in the local center gadget next to it has only a $5 \times 5$ room on the other side (unless two end gadgets are directly connected, which is not useful). Depending on where we place it, more or less of the room is available to be included in neighboring galaxies.

\begin {observation} \label {obs:end}
  If an end gadget is connected to a local center gadget, and on the other side of the local center gadget is not another end gadget, then we must place at least one galaxy center inside the end gadget.
\end {observation}

\subsection {High-level gadgets}

With these low-level gadgets in place, we now turn to constructing high-level gadgets that will encode the variables and clauses of our 1-in-3-SAT instance.

\textbf{Fix gadgets.}
First, we define the \defn{fix gadget}: an alternating sequence of four local center gadgets and three block gadgets making a U-turn as in Figure~\ref {fig:fix-gadget}.
The purpose of this gadget is to ensure that, in an optimal solution, each $5\times5$ block will be completely included in a single galaxy, and not split among several galaxies.

\begin {lemma} \label {lem:fix}
In any chain of blocks that contains a fix gadget, an optimal solution requires that each block must completely belong to a single galaxy.
\end {lemma}
\begin {proof}
First, by Observation~\ref {obs:desire}, an optimal solution has exactly one center in each local center gadgets, and none in the block gadgets.
This implies that if {\em any} block gadget in the chain does not completely belong to a single galaxy, then, by symmetry, {\em every} block gadget in the chain does not completely belong to a single galaxy.

Now, suppose we have a fix gadget consisting of blocks $A$, $B$, and $C$, where $A$ and $C$ are corner blocks but $B$ is a straight block, and suppose $A$ is split among multiple galaxies, say a red one at the top and a blue one at the right; refer to Figure~\ref {fig:fix-gadget}(b). Then, by symmetry of the blue galaxy, the center left pixel of $B$ must also be blue and the bottom center pixel of $B$ must also be red. Then, the top center pixel of $C$ must be red (a priori, this could be the same or a different red galaxy --- by connectivity it must be a new galaxy) and the center right pixel of $C$ must be blue (again, either the same or a different blue galaxy) by symmetry of the (second) red galaxy. But when both the center left and top center pixels of $C$ are red then $C$ must be completely red, a contradiction.
\end {proof}

 \begin{figure}
\centering
\includegraphics [scale=0.75]{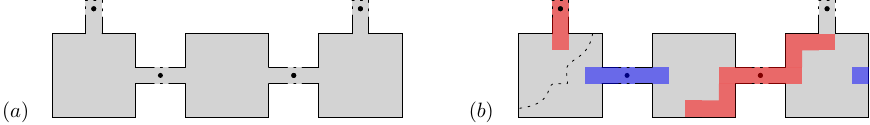}
  \caption{\small (a) The fix gadget. (b) If the left block is split among multiple galaxies: contradiction.}
  \label{fig:fix-gadget}
\end{figure}

 \begin{figure}
\centering
\includegraphics [scale=0.75]{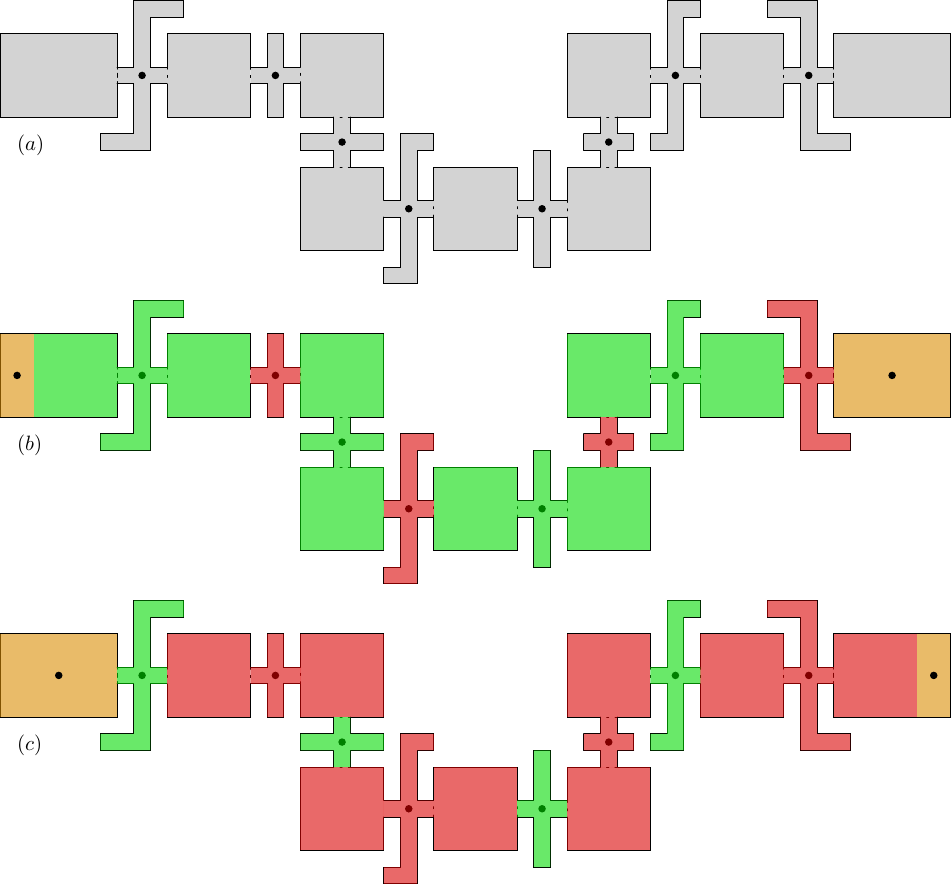}
  \caption{\small A variable chain and its two possible truth assignments.}
  \label{fig:variable-chain}
\end{figure}

\textbf{Variable gadgets.}
A \defn{variable chain} consists of alternating local center gadgets and block gadgets starting and ending at end gadgets; see Figure~\ref{fig:variable-chain} for a variable chain.
Each variable chain must include at least one fix gadget.
Then, by Lemma~\ref {lem:fix}, each block gadget on a variable chain must belong to exactly one of its two neighboring local center gadgets; this choice propagates throughout the entire chain and encodes the value \textsc{true} or \textsc{false} of one of the variables in our 1-in-3-SAT instance.
Specifically, we show:

\begin {lemma} \label {lem:variable}
Consider a variable chain which
\begin {itemize}
  \item has $k>1$ local center gadgets and $k-1$ block gadgets;
  \item contains at least one fix gadget; and
  \item has the property that every three consecutive local center gadgets have three different shapes.
\end {itemize}
Then there exists no set of $k+1$ or fewer galaxy centers that form a valid puzzle, and there are exactly two sets of $k+2$ galaxy centers that form a valid puzzle, both of which include all $k$ \emph{desired} galaxy centers.
\end {lemma}
\begin {proof}
By Observation~\ref {obs:end} and because $k > 1$, we must place one galaxy center inside each end gadget.
By Observation~\ref {obs:desire} and because we vary the shapes of the local center gadgets, we must place at least one galaxy center at each of the $k$ local center gadgets.
Indeed, a galaxy centered in a local center gadget cannot include both adjacent local center gadgets, since they have different shapes, and a galaxy centered outside a local center gadget cannot include both incident local center gadgets, since they have different shapes.
So, we need at least $k+2$ centers.
This is also sufficient, as shown by the two solutions with $k+2$ centers in Figure~\ref {fig:variable-chain}.
Finally, by Lemma~\ref {lem:fix}, we cannot have more than two possible shapes for each galaxy.
\end {proof}

Note that within one variable chain, we can realize the conditions of Lemma~\ref {lem:variable} with only three different local center gadget shapes, which we alternate cyclicly.

 \begin{figure}
\centering
\includegraphics [scale=0.75]{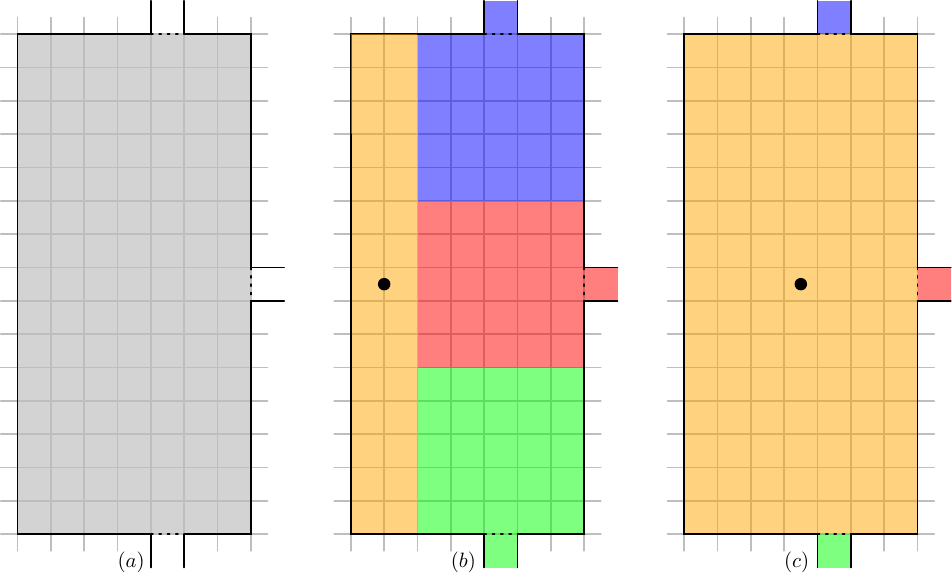}
  \caption{\small The split gadget is essentially three end gadgets, but can only be solved with a single center if either all three blocks are used or all three are not used.}
  \label{fig:split-gadget}
\end{figure}

\textbf{Split gadgets.}
In order to connect variables to multiple clauses, we need to be able to split them.
Rather than splitting chains in the usual sense, our \defn{split gadget} takes three independent variable chains and forces them to have the same state by combining three end gadgets into a single larger end gadget; see Figure~\ref {fig:split-gadget}. This effectively creates a ``split'' chain with three remaining ends that have the same state; by repeating this we may create as many copies of a variable as required.

The idea behind the split gadget is simple:
by combining three end gadgets into one large $7\times15$ room, we no longer know that we need three separate galaxy centers; in fact, applying Observation~\ref {obs:end} to the three adjacent local center gadgets, we potentially only require to place one galaxy center instead of three.
Observe that this is indeed achievable, but \emph{only} when all three chains are in the same state, because otherwise the remaining shape of the room is not symmetric by $180^\circ$ rotation (and thus cannot be covered by a single additional galaxy).

\begin{observation}
  A split gadget requires only one galaxy center if and only if either all three adjacent local center gadgets do use a $5\times5$ area of its room, or if none of them do.
\end{observation}

\textbf{Clause gadgets.}
A clause of our 1-in-3-SAT formula is now simply encoded by a single clause block where three variable chains meet: it can be solved without using an additional center if and only if precisely one of the three chains needs to use the block.

In order to avoid that galaxies centered on clause blocks include local center gadgets, the three local center gadgets directly incident to the clause block must have distinct shapes.
To achieve this, we may start with an edge coloring of the variable--clause incidence graph and use three unique shapes for each color. Every planar graph of maximal degree $3$ admits an edge coloring with 4 colors, which results in 12 distinct shapes for our local center gadgets.
It is possible, but not necessary, to reduce the number of shapes further.

\subsection{Global construction}

With our high-level gadgets now also in place, we can describe the global construction of our reduction.
We start from an planar embedding of $G$
with edges routed orthogonally on an $O(n) \times O(n)$ grid
\cite{biedl98}, scaled up by a constant factor.
We locally replace each clause by a single clause block, each variable by a sufficient number of split gadgets connected by variable chains, and each edge by variable chains connecting variables to clauses.


\textbf{Distance considerations.}
We must make sure that all our gadgets are aligned with the grid, and can be connected to each other. Furthermore, we need to ensure that the number of block gadgets in each variable chain
has the same parity so that all clauses use positive forms of the variable.

We observe that we can achieve all of these requirements by simply adjusting the distance between adjacent block gadgets and allowing a slightly richer class of local center gadgets: instead of a width of $3$ we may give them a width of any desired integer $w$. If $w$ is even, we must place a galaxy center on an edge and not in the center of a cell, and we correspondingly also adjust the vertical arms of the gadget by giving them a width of $2$ rather than $1$.

This concludes the proof of Theorem~\ref{th:nrc}.
%

%% file: puzz.tex
\begin{figure}
\centering
\comicII{.8\textwidth}{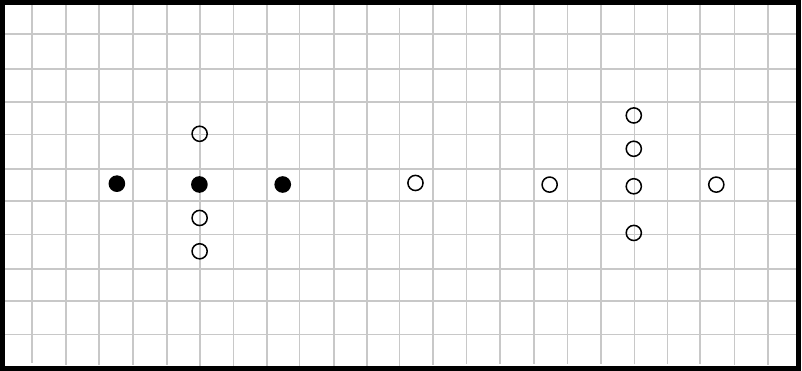}
  \caption{\small Puzzle that can be solved for letters A, B, H, R, S, Z (E for disconnected galaxies). 
}
  \label{fig:letter-puzzle}
\end{figure}

Standard Spiral Galaxies puzzles are designed to have a unique solution, with the intent that the puzzle solver finds that one solution.
In this section, we briefly explore the possibility of creating a single puzzle with \emph{multiple} solutions, such that each solution forms a different desired picture.
As a case study, we ask for a single puzzle such that, for every letter of the alphabet, there exists a solution to the puzzle that resembles that letter.
In Figure~\ref{fig:letter-puzzle}, we provide a puzzle which has a solution for the black letters A, B, H, P, R, S, Z (and when we allow for \emph{disconnected} galaxies also for the letter E). 
See Figure~\ref{fig:letter-puzzle-sol} for the solutions.

\begin{figure}
\centering
\comic{.45\textwidth}{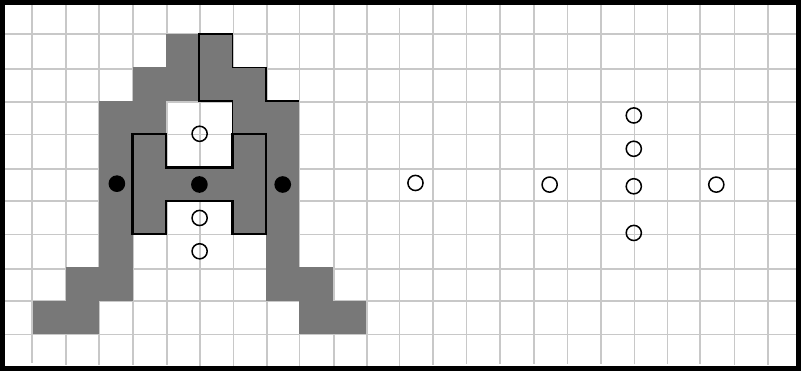}{(a)}
\comic{.45\textwidth}{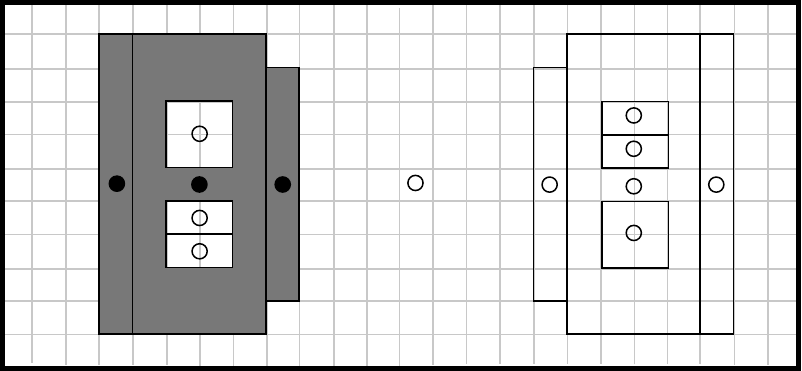}{(b)}
\comic{.45\textwidth}{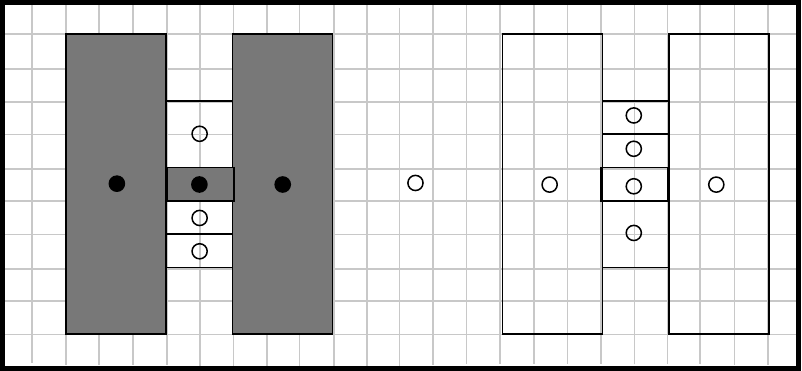}{(c)}
\comic{.45\textwidth}{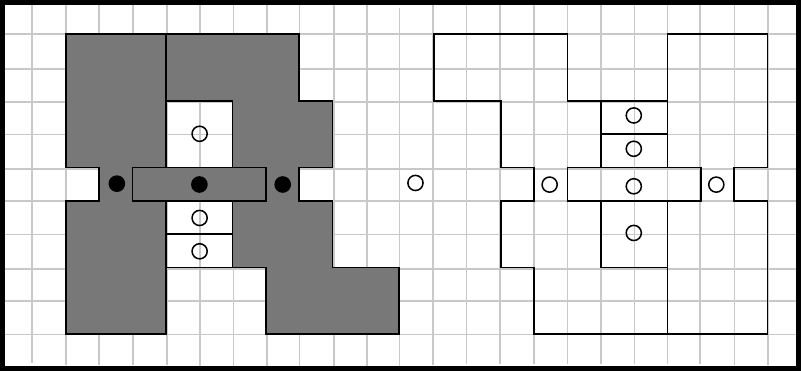}{(d)}
\comic{.45\textwidth}{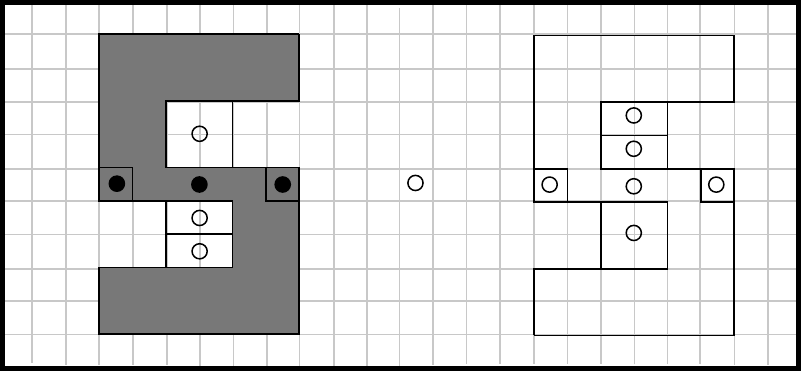}{(e)}
\comic{.45\textwidth}{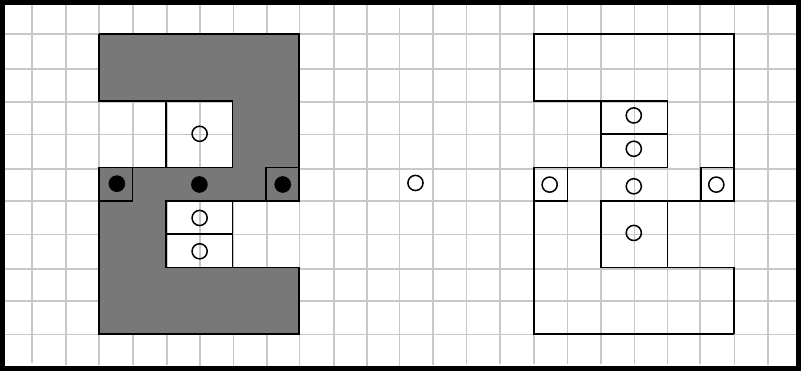}{(f)}
\comic{.45\textwidth}{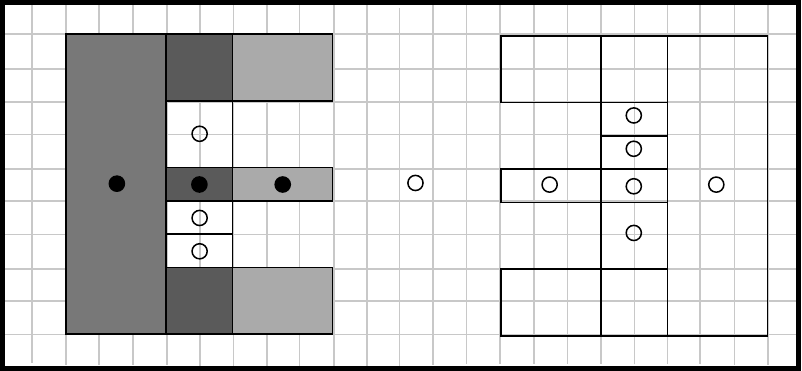}{(g)}
\comic{.45\textwidth}{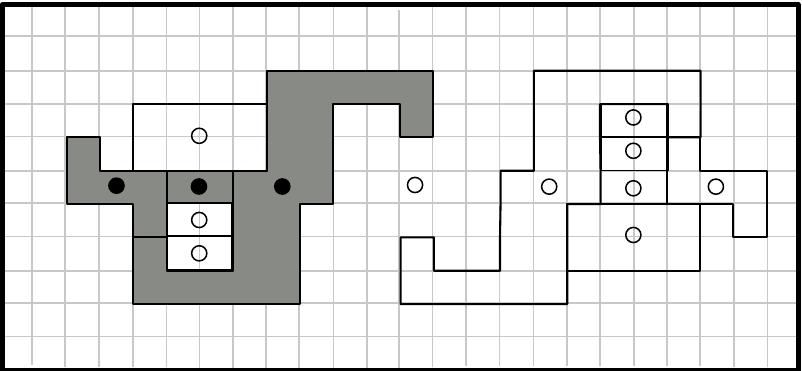}{(h)}
 \caption{\small (a--f) One puzzle with solutions drawing the letters A, B, H, R, S, and Z. (g) When we allow disconnected galaxies, we can also make the letter E. The three dark gray areas together form a single galaxy, and the three light gray areas as well. (h) The puzzle also has many ``nonsense'' solutions.}
  \label{fig:letter-puzzle-sol}
\end{figure}

While our construction has solutions for several letters, for most letters of the alphabet it does not, and it also has many additional solutions which do not resemble letters at all.
We leave as an open problem to construct puzzles which either solve for more letters, or for fewer nonletters, or both.
More generally, this raises the question of whether it is possible to create puzzles with multiple solutions in such a way that \emph{every} valid solution forms a meaningful picture.
